\documentclass[10pt,a4paper,reqno]{amsart}

\usepackage{latexsym,amssymb,amsmath,amsthm,amsfonts,enumerate,verbatim,xspace,
exscale}
\usepackage{graphicx}
\usepackage{color,amsbsy}
 
\usepackage[all]{xy}

\usepackage{hyperref}
\definecolor{darkblue}{rgb}{0,0,.5}
\hypersetup{colorlinks=true, breaklinks=true, linkcolor=darkblue, menucolor=darkblue,  urlcolor=darkblue,citecolor=darkblue}

\theoremstyle{plain}
\newtheorem{theorem}{Theorem}[section]

\newtheorem{proposition}[theorem]{Proposition}

\theoremstyle{definition}

\newtheorem{remark}[theorem]{Remark}

\def\C{\mathcal{C}}
\def\R{\mathbb{R}}

\def\F{\mathcal{F}}

\providecommand{\bysame}{\makebox[3em]{\hrulefill}\thinspace}

\newcommand{\longto}{\longrightarrow}

\def\vv<#1>{\langle#1\rangle}
\def\ww<#1>{\langle\langle#1\rangle\rangle}

\providecommand{\del}{\partial}

\newcommand{\om}{\omega}
\newcommand{\Om}{\Omega}

\newcommand{\eps}{\varepsilon}
\newcommand{\lam}{\lambda}
\newcommand{\Lam}{\Lambda}

\newcommand{\by}[2]{\mbox{$\frac{#1}{#2}$}}
\newcommand{\cinf}{\mbox{$C^{\infty}$}}
\providecommand{\set}[1]{\mbox{$\{#1\}$}}

\newcommand{\gu}{\mathfrak{g}}

\newcommand{\Ad}{\mbox{$\text{\upshape{Ad}}$}}
\newcommand{\ad}{\mbox{$\text{\upshape{ad}}$}}

\newcommand{\ds}{\mbox{$\,\delta_t$}}
\newcommand{\di}{\mbox{$\,d_t$}}
\newcommand{\cl}[1]{\mbox{$[#1]_{\mathfrak{g}^*}$}}
\newcommand{\bigcl}[1]{\mbox{$\Big[#1\Big]_{\mathfrak{g}_0^*}$}}
\newcommand{\SVect}{\mbox{$\textup{SVect}$}}
\newcommand{\Vect}{\mbox{$\textup{Vect}$}}

\title[Probabilistic representation of helicity in viscous fluids]{%
Probabilistic representation of helicity \\
in viscous fluids
}

\usepackage[foot]{amsaddr} 
\author{Simon Hochgerner}
\address{\"Osterreichische Finanzmarktaufsicht (FMA),
Otto-Wagner Platz 5, A-1090 Vienna
}
\email{simon.hochgerner@fma.gv.at} 

\begin{document}

\begin{abstract}
It is shown that the helicity of three dimensional viscous incompressible flow can be
identified with the overall linking of the fluid's initial vorticity to the expectation of a stochastic mean field limit.  
The relevant mean field limit is obtained by following the Lagrangian paths in the stochastic Hamiltonian interacting particle system of [S.\ Hochgerner, Proc.\ R.\ Soc.\ A 
\textbf{474}:20180178]. 
\end{abstract}

\maketitle

\section{Introduction} 


The evolution of the velocity field, $u=u(t,x)=u_t(x)$, of a three dimensional incompressible fluid with constant mass density, $\rho=1$, is given by 
\begin{equation}
    \label{1e:NS}
    \by{\del}{\del t}u
    = 
    -\nabla_u u - \nabla p + \nu\Delta u,
    \qquad
    \textup{div}\,u = 0,
    \qquad
    u(0,.) = u_0
\end{equation}
where $t\in[0,T]$, $x\in \R^3$, $\nabla_u u = \vv<u,\nabla>u = \sum_{j=1}^3 u^j\del_j u$ is the covariant derivative in $\R^3$, $p = p(t,x)$ is the pressure determined by $\textup{div}\,u=0$, and the smooth divergence free vector field $u_0$ is an initial condition which decays sufficiently rapidly at infinity.    
If $\nu>0$ then \eqref{1e:NS} is the incompressible Navier-Stokes equation, and if $\nu=0$ it is the incompressible Euler equation in $\R^3$. 

The helicity of the fluid is defined by 
\begin{equation}
    \label{1e:hel}
    \mathcal{H}_t
    = \int_{\mathbb{R}^3} \vv<u_t,\textup{curl}\,u_t>\,dx
\end{equation}
where $dx$ is the Euclidean volume element in $\R^3$. Helicity is a topological quantity measuring the overall degree of linking and knotting of vortex lines (Moffatt et al.~\cite{M69,MR92,MT92}). 

Let $u$ be a solution to the Euler equation ($\nu=0$) and consider the Lagrangian flow $g=g_t(x)$ generated by 
\[
 \by{\del}{\del t}g_t = u_t\circ g_t,
 \qquad
 g_0 = e
\] 
where $e$ is the identity map in $\R^3$.
Then $g_t$ is a curve in the group of volume preserving diffeomorphisms $\textup{SDiff}(\R^3)$. Arnold~\cite{A66} has shown that the Euler equation has the structure of an infinite dimensional Hamiltonian system with configuration space $\textup{SDiff}(\R^3)$.
Moreover, the Hamiltonian function, $\int_{\mathbb{R}^3}\vv<u,u>\,dx\,/2$, is invariant under the relabeling symmetry, which is given by composition from the right in $\textup{SDiff}(\R^3)$.
Thus Noether's theorem applies, and yields
\begin{equation}\label{1e:cons}
    u_t
    = \Ad(g_t^{-1})^{\top}u_0
\end{equation}
where the transpose adjoint action, $\Ad(\cdot)^{\top}$, is defined as follows: For a vector field, $X$, let $PX = X - \nabla\Delta^{-1}\textup{div}\,X$ be the Leray-Hodge projection onto the divergence free part. Then  $\Ad(h)^{\top}v = P\cdot(Th)^{\top}\cdot(v\circ h)$ where $h\in\textup{SDiff}(\R^3)$, $(Th)^{\top}$ is the transpose matrix, and $v$ is a divergence free vector field. 
In fact, $\Ad(h)^{\top}$ is the transpose with respect to the $L^2$ inner product to the adjoint action (inverse vector field pullback) $\Ad(h): v\mapsto Th\cdot(v\circ h^{-1})$.
We remark that $Th = (\del_i h^j)_{j,i}$ will throughout refer to differentiation in the space variable, while time differentiation will be denoted by $\frac{\del}{\del t}$ (ordinary), $\ds$ (Stratonovich), or $\di$ (Ito calculus). 

The transport equation~\eqref{1e:cons} and the identity $\textup{curl}\,\Ad(g^{-1})^{\top} = \Ad(g)\,\textup{curl}$ yield
\[
 \mathcal{H}_t
 = \int_{\mathbb{R}^3} \Big\langle
  \Ad(g_t^{-1})^{\top}u_0 , \textup{curl}\,\Ad(g_t^{-1})^{\top}u_0
  \Big\rangle \,dx
  = \int_{\mathbb{R}^3} \Big\langle
  u_0 , \Ad(g_t^{-1})\,\textup{curl}\,\Ad(g_t^{-1})^{\top}u_0
  \Big\rangle \,dx
  = 
  \mathcal{H}_0. 
\]
Hence Euler flow conserves helicity (\cite{M61,M69}).

In the viscous case, $\nu>0$, helicity is generally not conserved. The precise mechanism whereby helicity changes under Navier-Stokes flow is subject to ongoing investigation  (\cite{Scheeler,LRS15,Kerr15,Kerr17}).  

This note studies helicity in the viscous case from the point of view of stochastic Hamiltonian interacting particle systems (SHIPS) as in \cite{H17,H18,H20}. These systems can be viewed as a stochastic perturbation (along Hamiltonian vector fields) of ideal fluid mechanics. Ideal fluid mechanics (Euler flow) preserves energy and helicity. Both are quadratic invariants, but they have different geometric origins. Energy conservation follows because the Hamiltonian coincides with the energy functional, but such a quantity is generally not preserved under stochastic Hamiltonian perturbations. On the other hand, helicity is a Casimir function (constant on coadjoint orbits), and Casimirs are preserved by stochastic Hamiltonian perturbations. Hence the SHIPS approach can be expected to possess a helicity type invariant. 

In \cite{H18} it is shown that solutions of incompressible Navier-Stokes equation can be obtained as the mean field limit of these interacting particle systems, and one may wonder how the helicity preserving stochastic Hamiltonian construction gives rise to a flow with non-constant helicity. 
The observation of this paper is that helicity in viscous fluids may be considered as an average over cross-helicities of stochastically perturbed ideal flows. Moreover, the group structure in $\textup{SDiff}(\R^{3})$ allows to identify this as the cross-helicity of initial vorticity and an average over backward-forward transports of the initial velocity.

Concretely, and to describe the stochastic Hamiltonian equations in question, fix a (large) integer $N$ and consider for, $\alpha = 1,\dots,N$, the IPS
\begin{align}
\label{1e:syst}
    (\ds g_t^{\alpha})\circ (g_t^{\alpha})^{-1}
    &= \frac{1}{N}\sum_{\beta=1}^N u_t^{\beta} \ds t + \sqrt{2\nu}\ds W_t^{\alpha} 
    ,\qquad g_0^{\alpha} = e
    , \qquad
    u_t^{\alpha}
    = \Ad\Big( (g_t^{\alpha})^{-1} \Big)^{\top} u_0 
\end{align}
where $\ds$ denotes Stratonovich differentiation and $(W^{\alpha})$ is a sequence of $N$ mutually independent Brownian motions in $\R^3$.  
The process $g_t^{\alpha}$ takes values in the group of volume preserving diffeomorphisms, $\textup{SDiff}(\R^3)$, and $e$ is the identity diffeomorphism. 
The system~\eqref{1e:syst} is presented in Section~\ref{sec:ships} from the Lie-Poisson point of view. 

The process $u_t^{\alpha}$ takes values (by construction) in the space of divergence free vector fields, $\SVect(\R^3)$. Further, it depends on $N$ and we can consider the mean field limit, $u_t^{\infty} = \lim_{N\to\infty}u_t^{\alpha}$, which is a limit in probability (\cite{Oel84,DV95,JW17}). In fact, since all particles are identical, it suffices to consider the limit for $\alpha=1$. 
Theorem~\ref{thm:prel}, which is a summary of \cite{H18}, shows that a given solution, $u_t$, to the Navier-Stokes equation can be represented as $u_t = E[u_t^{\infty}] = \lim_{N\to\infty}\sum_{\alpha=1}^N u_t^{\alpha}/N$. 
Therefore, and also by analogy to the ideal fluid case \eqref{1e:cons}, it makes sense to call $\Ad( (g_t^{\alpha})^{-1} )^{\top} u_0$ the forward transport of $u_0$. That is, the initial condition is transported forward in time along the stochastic Lagrangian path $g_t^{\alpha}$.  
The IPS~\eqref{1e:syst} arises from a decomposition of each infinitesimally small blob of fluid (at each $x\in \R^3$) into $N$ identical sub-blobs, and insisting that the sub-blobs follow their common center of mass while at the same time undergoing each their own stochastic process. See Section~\ref{sec:2phys}.

Fix an index $\alpha$ and consider the fluid collection that is made up of all $\alpha$-sub-blobs. In Section~\ref{sec:hel-rep-ships} it is observed that helicity is indeed preserved along the corresponding stochastic flow~\eqref{1e:syst}, that is 
\[
    \int_{\mathbb{R}^3} \vv<u_t^{\alpha},\textup{curl}\,u_t^{\alpha}>\,dx
    = \mathcal{H}_0
\]
for all $\alpha=1,\dots,N$.
This yields the representation
\begin{equation}
    \tag{Theorem~\ref{thm:hel}}
 \mathcal{H}_t
 =
  \lim_{N\to\infty}\sum_{1\le\alpha\neq\beta\le N}\int_{\mathbb{R}^3}
    \Big\langle
     \Ad(g_t^{\alpha})^{\top} \Ad((g_t^{\beta})^{-1})^{\top} \,u_0 ,
    \textup{curl}\,u_0 
    \Big\rangle\,dx\, / N^2 
\end{equation}
for the helicity~\eqref{1e:hel} of Navier-Stokes flow (see also Remark~\ref{rem:hel}).
Hence the initial velocity, $u_0$, is transported forward  along a stochastic Lagrangian path, $g_t^{\beta}$, and then backwards along another path, $g_t^{\alpha}$. 
The result, $\mathcal{H}_t$, is obtained by averaging over the $L^2$ inner products of the initial vorticity and all such backward-forward transports with $1\le\alpha\neq\beta\le N$, and letting $N$ tend to infinity. 
Helicity at time $t$ is thus the average over all possible linkings of integral curves, corresponding to $\textup{curl}\,\Ad(g_t^{\alpha})^{\top} \Ad((g_t^{\beta})^{-1})^{\top} \,u_0$, and initial vortex lines, corresponding to $\textup{curl}\,u_0$. See also Section~\ref{sec:3phys}. 

If we set $\nu=0$, the Lagrangian paths $g_t^{\alpha}$ are deterministic and coincide with each other, whence it follows that 
$\Ad(g_t^{\alpha})^{\top} \Ad((g_t^{\beta})^{-1})^{\top} \,u_0 = u_0$ and
Theorem~\ref{thm:hel}
reduces to $\mathcal{H}_t = \int_{\mathbb{R}^3}\vv<u_0,\textup{curl}\,u_0>\,dx = \mathcal{H}_0$,   
which is the conservation of helicity in the inviscid case. 

Consider the SDE that follows from \eqref{1e:syst} in the mean field limit as $N\to\infty$. To obtain this limit we may fix $\alpha=1$ since all interacting particles (i.e., sub-blobs) are identical. The result is the stochastic mean field system
\begin{equation}
    \label{1e:syst_mf}
    (\ds g_t^{\infty})\circ(g_t^{\infty})^{-1}
    = E[u_t^{\infty}]\ds t + \sqrt{2\nu}\ds W_t,
    \qquad
    g_0^{\infty} = e,
    \qquad 
    u_t^{\infty} = \Ad\Big( (g_t^{\infty})^{-1} \Big)^{\top}u_0 
\end{equation}
for processes $g_t^{\infty}$ in $\textup{SDiff}(\R^3)$ and $u_t^{\infty}$ in $\SVect(\R^3)$, and where $W$ is Brownian motion in $\R^3$. It follows that $u_t = E[u_t^{\infty}]$ satisfies the Navier-Stokes equation (Theorem~\ref{thm:prel}).

The IPS approach in \cite{H17,H18,H20} is a Hamiltonian analogue of the Constantin and Iyer~\cite{CI05} representation of solutions to the Navier-Stokes equation via the stochastic Weber formula. In fact, \eqref{1e:syst_mf} is equivalent to \cite{H18} via the stochastic Noether theorem (cf.\ Theorem~\ref{thm:prel}), and coincides, up to notation, with the stochastic Weber formula of \cite[Theorem~2.2]{CI05}. In the present context the $\Ad(\cdot)^{\top}$ notation is kept because the group theoretic formulation is helpful in the helicity calculations. 
 
The system~\eqref{1e:syst_mf} leads to an expression for the mean field limit in 
Theorem~\ref{thm:hel},
which is 
\begin{equation}
    \tag{Theorem~\ref{thm:hel-mf}}
    \mathcal{H}_t
    =
    E\Big[
     \int_{\mathbb{R}^3}\vv< \Ad(h_t)^{\top}u_0, \,\textup{curl}\,u_0 >\, dx 
    \Big]
\end{equation}
where the process $h_t$ in $\textup{SDiff}(\R^3)$ is the solution to the SDE with random coefficients
\[
 (\ds h_t)\circ h_t^{-1}
 = \sqrt{2\nu}\, (Tg_t^{\infty})^{-1}\cdot \ds(B_t - W_t),
 \qquad 
 h_0 = e
\] 
and where $B$ is a Brownian motion in $\R^3$ which is independent of $W$. 
Section~\ref{sec:3phys} contains a physical interpretation of these equations.

These results depend on the existence of the mean field limits under consideration. The existence of these limits is assumed, at least for a short period of time $[0,T]$, but not proven (in this paper).  Constantin and Iyer~\cite{CI05} have shown short time existence for \eqref{1e:syst_mf}.

\section{Stochastic Hamiltonian interacting particle system (SHIPS) and mean-field limit}\label{sec:ships}
The stochastic Hamiltonian approach that is presented in this section has been developed in \cite{H17,H18,H20}. 
This approach is a Hamiltonian analogy to the mean field Weber formula theory of Constantin and Iyer~\cite{CI05}, and it is also related to Holm's variational principle for stochastic fluid mechanics (\cite{Holm15}). 
Section~\ref{sec:vor} contains the mean field evolution equation for stochastic vorticity, which is a straightforward consequence of \cite{H18} but has not been presented in this form elsewhere. A brief explanation of the physical picture underlying the mean field approach is given in Section~\ref{sec:2phys}.

\subsection{Diffeomorphism groups}\label{sec:diffgps}
We fix $s>5/2$ and let $\textup{SDiff}({\mathbb{R}^3})$ denote the infinite dimensional $\cinf$-manifold of volume preserving $H^s$-diffeomorphisms on ${\mathbb{R}^3}$. 
This space is a topological group, but not a Lie group since left composition is only continuous but not smooth. Right composition is smooth.
Let 
\[
 \gu = \SVect({\mathbb{R}^3})
\]
be the space of divergence free vector fields on ${\mathbb{R}^3}$ of class $H^s$.
The tangent space of $\textup{SDiff}({\mathbb{R}^3})$ at the identity $e$ consists of divergence free and compactly supported vector fields, denoted by
\[
 T_e \textup{SDiff}({\mathbb{R}^3})
 = \gu_0
 = \SVect({\mathbb{R}^3})_{\textup{cp}}.
\]
We use right multiplication $R^g: \textup{SDiff}({\mathbb{R}^3})\to \textup{SDiff}({\mathbb{R}^3})$, $k\mapsto k\circ g = kg$ to trivialize the tangent bundle $T\textup{SDiff}({\mathbb{R}^3})\cong \textup{SDiff}({\mathbb{R}^3})\times\gu_0$, $v_g\mapsto(g,(TR^g)^{-1} v_g)$.

The $L^2$ scalar product $\ww<.,.>$ on $\gu_0$ is defined by
\[
  \ww<v,w> 
  = \int_{\mathbb{R}^3}\vv<v(x),w(x)>\, dx
\]
for $v,w\in\gu_0$, where $dx$ is the standard volume element in ${\mathbb{R}^3}$, and $\vv<.,.>$ is the Euclidean inner product.  
Via $R^g$ this can be extended to a right invariant Riemannian metric on $\textup{SDiff}({\mathbb{R}^3})$. 
See \cite{AK98,EM70,MEF,Michor06}.

\subsection{Phase space}\label{sec:PS}
The configuration space of incompressible fluid mechanics on ${\mathbb{R}^3}$ is $\textup{SDiff}({\mathbb{R}^3})$. 
The corresponding phase space is trivialized via right multiplication as
\[
 T^*\textup{SDiff}({\mathbb{R}^3})  \cong \textup{SDiff}({\mathbb{R}^3}) \times   \gu^*
\]
where $\gu^*$ is defined as $\gu^* = \Om^1(\mathbb{R}^3)/d\F(\mathbb{R}^3)$.
Here $\Om^k(\mathbb{R}^3)$ are $k$-forms (of class $H^s$), $\mathcal{F}(\mathbb{R}^3)$ are functions (of class $H^{s+1}$), and $\Om^1(\mathbb{R}^3)/d\F(\mathbb{R}^3)$ is the space of equivalence classes modulo exact one-forms. 
Elements in $\gu^*$ will thus be denoted by $\cl{\xi}$ where $\xi\in\Om^1(\mathbb{R}^3)$ is a representative of the class in $\gu^* = \Om^1(\mathbb{R}^3)/d\F(\mathbb{R}^3)$. 

Let $\flat: \Vect(\mathbb{R}^3)\to\Om^1(\mathbb{R}^3)$, $X\mapsto X^{\flat}$ be the metric (musical) isomorphism with inverse $\flat^{-1}=\sharp$. Let $P: \Vect(\mathbb{R}^3)\to\SVect(\mathbb{R}^3)$, $X\mapsto X-\nabla\Delta^{-1}\textup{div}(X)$ be the Hodge projection onto divergence free vector fields. Then we obtain an isomorphism $\mu: \gu\to\gu^*$, $X\mapsto\cl{X^{\flat}}$ with inverse $\mu^{-1}: \gu^*\to\gu$, $\cl{\xi}\mapsto P \xi^\sharp$.

\begin{remark}\label{rem:no-iso}
The restriction of $\mu$ to $\gu_0$ does not induce an isomorphism of $T\textup{SDiff}(\mathbb{R}^3)\cong\textup{SDiff}(\mathbb{R}^3)\times\gu_0$ and $T^*\textup{SDiff}(\mathbb{R}^3)\cong \textup{SDiff}(\mathbb{R}^3) \times \gu^*$ since $\mu^{-1}(\cl{\xi})$ need not be compactly supported.
\end{remark}

\subsection{SHIPS}
In \cite{H18} a stochastic Hamiltonian interacting particle system is constructed which yields, in the mean field limit, the solution to the incompressible Navier-Stokes equation. Let $N$ be the number of interacting particles (or interacting blobs of fluid). The approach is 
Hamiltonian and the relevant phase space is
\[
 \mathcal{P}^N
 = 
 \Big( \textup{SDiff}(\mathbb{R}^3)\times\gu^*\Big)^N.
\]
This space is equipped with the direct product symplectic structure obtained from the canonical symplectic form on each copy  $\textup{SDiff}(\mathbb{R}^3)\times\gu^*$. 
The direct product group $\textup{SDiff}(\mathbb{R}^3)^N$ acts on $\mathcal{P}^N$ through the product action and there is a corresponding momentum map
\begin{equation}
    \label{e:momap}
    J: \mathcal{P}^N\to(\gu^*)^N,\qquad
    \Big(g^{\alpha},\cl{\xi^{\alpha}}\Big)_{\alpha=1}^N
    \mapsto 
    \Big(
    \Ad(g^{\alpha})^*\cl{\xi^{\alpha}} \Big)_{\alpha=1}^N
\end{equation}
where the coadjoint representation is determined by 
$\vv<\Ad(g)^*\cl{\xi}, X> = \int_{\mathbb{R}^3}\vv<\xi,\Ad(g)X>\,dx$ 
and $\Ad(g)X = Tg\cdot(X\circ g^{-1}) = (g^{-1})^*X$ 
for $g\in\textup{SDiff}(\R^3)$,  $\cl{\xi}\in\gu^*$ and $X\in\gu_0$. 
That is, 
\[
 \Ad\Big(g^{\alpha}\Big)^*\bigcl{\xi^{\alpha}} 
 = \bigcl{(g^{\alpha})^*\xi^{\alpha}}
 = \bigcl{ (\xi^{\alpha}\circ g^{\alpha})\cdot Tg^{\alpha} }
 .
\]
The infinitesimal adjoint representation is given by $\ad(X).Y = [X,Y] = -L_X Y$ where $L_X Y = \nabla_X Y - \nabla_Y X$ is the Lie derivative. The corresponding coadjoint representation $\ad(X)^*: \gu^*\to\gu^*$ is characterized by $\ww<\ad(X)^*\cl{\xi},Y> = \ww<P\xi^{\sharp}, [X,Y]>$. Thus, $\ad(X)^*\cl{\xi} = \cl{L_X\xi}$ where $L_X\xi$ is the Lie derivative of a one-form.

Let $e_1$, $e_2$, $e_3$ be the standard basis vectors in $\R^3$. In the following, the vectors $e_j$ will be viewed as constant  vector fields on $\mathbb{R}^3$. 
Let $(\Om,\F,(\F_t)_{t\in[0,T]},\mathbb{P})$ be a filtered probability space satisfying the usual assumptions as specified in \cite{Pro}. All stochastic processes shall be understood to be adapted to this filtration. For $\alpha=1,\dots,N$ consider a sequence of mutually independent Brownian motions $W^{\alpha} = \sum W^{j,\alpha}e_j$ in $\R^3$.  
The Stratonovich differential will be denoted by $\ds$. 

The equations of motion for a path $(g_t^{\alpha},\cl{\xi_t^{\alpha}})_{\alpha=1}^N$ in $\mathcal{P}^N$ are given by the system of Stratonovich SDEs  (see \cite[Equ.~(2.11)-(2.12)]{H18}):
\begin{align}
    \label{e:eom1}
    \ds g_t^{\alpha}
    &= TR^{g_t^{\alpha}}\Big(\frac{1}{N}\sum_{\beta=1}^N P(\xi_t^{\beta})^{\sharp} \ds t
            + \eps\sum_{j=1}^3 e_j\ds W^{j,\alpha} \Big),
    \qquad 
    g_0^{\alpha} = e \\
\label{e:eom2}
    \ds\bigcl{\xi_t^{\alpha}} 
    &=
    -\ad\Big(\frac{1}{N}\sum_{\beta=1}^N P(\xi_t^{\beta})^{\sharp} \Big)^*\bigcl{\xi_t^{\alpha}} \ds t
    - \eps\sum_{j=1}^3\ad\Big(e_j \Big)^*\bigcl{\xi_t^{\alpha}} \ds W^{j,\alpha},
    \qquad
    \xi_0^{\alpha} = u_0^{\flat}
\end{align}
where $\ds$ indicates Stratonovich differentiation, $\eps>0$ is a constant, $e$ is the identity diffeomorphism and $u_0\in\gu$ is a smooth deterministic and divergence free vector field. This system depends on the empirical average $\frac{1}{N}\sum_{\beta=1}^N P(\xi_t^{\beta})^{\sharp}$ and is therefore an interacting particle system. In the following we assume that the mean field limit of this IPS exists such that the limit in probability,
$
 \lim_{N\to\infty} \frac{1}{N}\sum_{\beta=1}^N P(\xi_t^{\beta})^{\sharp}
$,
is a deterministic time-dependent vector field and satisfies the desired initial condition.  
Moreover, for each $\alpha$, the process $\cl{\xi_t^{\alpha}}$ converges, as $N\to\infty$, to a stochastic process $\cl{\xi_t}$. Since the particles (i.e., fluid blobs) are identical, it suffices to consider $\cl{\xi_t^1}$, that is $\cl{\xi_t} = \lim_{N\to\infty}\cl{\xi^{1}_t}$. It follows that 
\begin{equation}
    \lim_{N\to\infty} \frac{1}{N}\sum_{\beta=1}^N P(\xi_t^{\beta})^{\sharp}
    = P\, E[\xi_t]^{\sharp}
    =: u_t.
\end{equation}

\begin{remark}\label{rem:N}
The process $(g_t^{\alpha},\cl{\xi_t^{\alpha}})$ depends on $N$. It would thus be more concise to write 
\[
 (g_t^{\alpha,N},\cl{\xi_t^{\alpha,N}})
\]
such that $\cl{\xi_t} = \lim_{N\to\infty}\cl{\xi^{1,N}_t}$.
However, to make the notation more readable the superscript $N$ is omitted, but it is always tacitly implied. A solution to \eqref{e:eom1}-\eqref{e:eom2} will mean a  strong solution on an interval $[0,T]$ independent of $N$ and such that the mean field limit exists. 
See \cite{Oel84,AD95,DV95,JW17} for background on mean field SDEs.
\end{remark}

\begin{remark}\label{rem:weak-sp}
The canonical symplectic form on $\textup{SDiff}(\mathbb{R}^3)\times\gu^*$ is only weakly symplectic. This follows from Remark~\ref{rem:no-iso} and implies that the induced homomorphism $T(\textup{SDiff}(\mathbb{R}^3)\times\gu^*)\to T^*(\textup{SDiff}(\mathbb{R}^3)\times\gu^*)$ is only injective but not surjective. Hence the Hamiltonian vector field does not exist for all functions on $\textup{SDiff}(\mathbb{R}^3)\times\gu^*$. 
In fact, equations~\eqref{e:eom1}-\eqref{e:eom2} do not arise from a Hamiltonian vector field since neither $\frac{1}{N}\sum_{\beta=1}^N P(\xi_t^{\beta})^{\sharp}$ nor $e_j$ are compactly supported. 
Not even the initial condition, $u_0$, is assumed to have compact support.
In \cite{H18} this problem was circumvented by taking the torus (which is compact) as the fluid's domain.
Therefore, in the present context, the Hamiltonian approach can be used only as a guiding principle. This means that, if $(g_t^{\alpha},\cl{\xi_t^{\alpha}})$ is a solution to \eqref{e:eom1}-\eqref{e:eom2}, the (Hamiltonian) conclusion $J(g_t^{\alpha},\cl{\xi_t^{\alpha}}) = J(g_0^{\alpha},\cl{\xi_0^{\alpha}})$ has to be proved directly.
Equation~\eqref{e:eom2} is a stochastic Euler equation, and so
restricting $P(\xi_t^{\alpha})^{\sharp}$ to be of compact support does not seem to be reasonable since solutions of the (deterministic) Euler equation are generally not expected to be compactly supported (\cite{CLV19}). 
\end{remark}

\begin{theorem}[\cite{H18}]\label{thm:prel}
Consider a solution $(g_t^{\alpha},\cl{\xi_t^{\alpha}})$ of \eqref{e:eom1}-\eqref{e:eom2}.
Then:
\begin{enumerate}
    \item 
    The equations for the stochastic mean field limit of the interacting particle system \eqref{e:eom1}-\eqref{e:eom2} are
    \begin{align}
    \label{e:mf-eom1}
    \ds g_t
    &= TR^{g_t}\Big( u_t \ds t
            + \eps\sum_{j=1}^3 e_j\ds W^j \Big),
    \qquad 
    g_0 = e \\
\label{e:mf-eom2}
    \ds\bigcl{\xi_t} 
    &=
    -\ad\Big( u_t\ds t + \eps\sum_{j=1}^3 e_j \ds W_t^j \Big)^*\bigcl{\xi_t} 
    \qquad
    \xi_0 = u_0^{\flat}
\end{align}
where $(g_t,\xi_t) = \lim_{N\to\infty}(g_t^1,\xi_t^1)$ and $u_t = E[P\xi_t^{\sharp}]$. 
    \item 
    $u_t
    = \lim_{N\to\infty} \frac{1}{N}\sum_{\beta=1}^N P(\xi_t^{\beta})^{\sharp}
    = E[P\xi_t^{\sharp}]$ solves the incompressible Navier-Stokes equation 
    \begin{equation}
        \label{e:NS}
             \by{\del}{\del t}u = -\nabla_u u - \nabla p + \nu\Delta u, \qquad \textup{div}\,u = 0
    \end{equation}
    where $p$ is the pressure and $\nu = \eps^2/2$.
    Conversely, if $u_t$ satisfies \eqref{e:NS} and $g_t$ is defined by \eqref{e:mf-eom1}, then $\cl{\xi_t} = \Ad(g_t^{-1})^*\cl{u_0^{\flat}}$ is a solution to \eqref{e:mf-eom2} and $E[P\xi_t^{\sharp}] = u_t$.
    \item
    Assume $(g_t^{\alpha})$ satisfies \eqref{e:eom1}.
    Then \eqref{e:eom2} holds if, and only if,
    $\Ad(g_t^{\alpha})^*\cl{\xi_t^{\alpha}} 
    = \cl{(g_t^{\alpha})^*\xi_t^{\alpha}} = \cl{u_0^{\flat}}$ for all $\alpha=1,\dots,N$.
    In the limit, $N\to\infty$, this implies $\Ad(g_t)^*\cl{\xi_t} = \cl{u_0^{\flat}}$. 
\end{enumerate}
\end{theorem}

\begin{proof}
These assertions are shown in \cite{H18} where item~(3) follows because \eqref{e:eom1}-\eqref{e:eom2} constitute a right invariant Hamiltonian system whence the momentum map \eqref{e:momap} is constant along solutions. In the present context (Remark~\ref{rem:weak-sp}) item~(3) is shown directly:\\
For an arbitrary $k$-form $\sigma$ we have the identity
$\ds (g_t^{\alpha})^*\sigma = (g_t^{\alpha})^* L_{(TR^{g_t^{\alpha}})^{-1}\delta_t g_t^{\alpha}} \sigma$.
Now,
$\Ad(g_t^{\alpha})^*\cl{\xi_t^{\alpha}} 
    = \cl{(g_t^{\alpha})^*\xi_t^{\alpha}} = \cl{u_0^{\flat}}$ for all $\alpha=1,\dots,N$ holds if, and only if,
\begin{align*}
    \ds \int_{\mathbb{R}^3}
    \vv< (g_t^{\alpha})^*\xi_t^{\alpha} , X >\, dx
    &=
    \int_{\mathbb{R}^3}
    \vv< 
    (g_t^{\alpha})^* L_{(TR^{g_t^{\alpha}})^{-1}\delta_t g_t^{\alpha}}  \xi_t^{\alpha} , X >\, dx
    +
    \int_{\mathbb{R}^3}
    \vv< (g_t^{\alpha})^*\ds\xi_t^{\alpha} , X >\, dx
    = 0
\end{align*}
for all $X\in\gu_0$.
Because of \eqref{e:eom1} the assertion follows. 
\end{proof}

\subsection{Vorticity formulation}\label{sec:vor}
The  system \eqref{e:eom1}-\eqref{e:eom2} is a stochastic version of ideal incompressible flow. Consequently, the corresponding vorticity may be expected to be transported along the stochastic flow. In this section it is shown that this is indeed the case. The vorticity, $\om=\om(\cl{\xi})$, associated to an element $\cl{\xi}\in\gu^*$ is defined as
\[
 \om = d\xi \in\C^2\subset\Om^2(\mathbb{R}^3)
\]
where $\C^2$ denotes the space of closed two-forms.
If $X = \mu^{-1}\cl{\xi} = P\xi^{\sharp}$ and $*$ is the Hodge star operator, then we have $(*\om)^{\sharp} = \nabla\times X$, which is the expression of the vorticity when considered as a vector field. We thus obtain an isomorphism, $\cl{\xi}\mapsto\om(\cl{\xi})$, from $\gu^*$ to $\C^2$. The induced coadjoint action on $\C^2$ is given by pullback, that is $\Ad(g)^*\om = g^*\om = (\om\circ g)\cdot\Lam^2 Tg$. The infinitesimal coadjoint action is given by the Lie derivative, $\ad(X)^*\om = L_X\om$.

Let $\cl{\xi_t^{\alpha}}$, for $\alpha=1,\dots,N$, be a solution to \eqref{e:eom2}.
The Maurer-Cartan formula, $L_X = di_X+i_Xd$, then implies that the vorticity, $\om^{\alpha}_t = d\xi^{\alpha}_t$, satisfies
\begin{align}
    \label{e:vor1}
    \ds\om^{\alpha}
    = 
    -\ad\Big(u^{(N)}\,\delta_t t +\eps\sum_{j=1}^3 e_j \ds W^{j,\alpha} \Big)^*\om^{\alpha}
    =
    -L_{u^{(N)}}\om^{\alpha}\ds t
    -\eps \sum_{j=1}^3 L_{e_j}\om^{\alpha}\ds W^{j,\alpha} 
\end{align}
where 
\[
 u^{(N)}
 = \sum_{\beta=1}^N P(\xi^{\beta})^{\sharp}/N
 = \sum_{\beta=1}^N BS \Big(( *\om^{\beta} )^{\sharp}\Big)/N ,
\]
and $( *\om^{\alpha} )^{\sharp}$ is the divergence free vector field associated to $\om^{\alpha}$ and $BS$ is the Biot-Savart operator. 
The latter is defined as 
\begin{equation}
    \label{e:BS}
    BS(w)(x)
    = \frac{1}{4\pi}\int_{\mathbb{R}^3}\frac{w(y)\times(x-y)}{(x-y)^3}\,dy
\end{equation}
for $w\in\gu=\SVect(\R^3)$.
Closed two-forms and divergence free vector fields are in one-to-one correspondence via $\om\mapsto(*\om)^{\sharp}$. 
This can be used to define $BS^*: \C^2\to\Om^1(\mathbb{R}^3)$, $\om\mapsto (BS((*\om)^{\sharp}) )^{\flat}$, which satisfies $BS^* d \xi = \xi$.

Applying $BS^*$ to \eqref{e:vor1} yields equation~\eqref{e:eom2}. Hence these equations are equivalent, and the former is the vorticity formulation of the latter. Furthermore, the constancy of the momentum map~\eqref{e:momap}, i.e.\ Theorem~\ref{thm:prel}(3),  along solutions implies that the vorticity is transported along the stochastic flow:
\begin{equation}
    \label{e:vor-trn}
    \Ad(g_t^{\alpha})^*\om_t^{\alpha} = \om_0
\end{equation}
for all $\alpha=1,\dots,N$. Note that the initial conditions are assumed to be independent of $\alpha$, $\om_0^{\alpha} = \om_0 = du_0^{\flat}$ and $g_0^{\alpha}=e$. 

Under the assumption that the mean field limit exists, we consider $\om_t = \lim_{N\to\infty}\om_t^1$. It follows that
\begin{align}
    \label{e:vor-mf}
    \ds \om_t
    &= -\ad\Big(u_t \ds t + \eps\sum_{j=1}^3e_j\ds W^j\Big)^*\om_t\\ 
    \notag
    u_t
    &=
    \Big(BS^*\Big(E[ \om_t]\Big)\Big)^{\sharp}
    =
    \lim_{N\to\infty} \Big(BS^*\Big(\sum_{\alpha=1}^N \om_t^{\alpha}/N \Big)\Big)^{\sharp}
\end{align}
which is a mean field SDE because the drift depends on the expectation. 
If $g_t = \lim g_t^1$ with $g_0=e$, this can be restated as $\om_t = \Ad(g_t^{-1})^*\om_0$.
Moreover, since $\cl{BS^*\om_t}$ satisfies \eqref{e:mf-eom2},
Theorem~\ref{thm:prel} implies that $u$ is a solution to the incompressible Navier-Stokes equation~\eqref{e:NS}.

\subsection{Physical interpretation of SHIPS}\label{sec:2phys}
The picture underlying the IPS \eqref{e:eom1}-\eqref{e:eom2} is that each infinitesimal blob of fluid is divided into $N$ identical sub-blobs. These sub-blobs interact to follow their common center of mass and, at the same time, each undergo their own Brownian motion. The barycentric component of the motion  is due to the $\sum_{\beta=1}^NP(\xi_t^{\beta})^{\sharp}/N$ part of the equation, while the stochastic perturbation is encoded in $\eps\sum e_j\ds W_t^{j,\alpha}$. 

In \cite{H18} the equations \eqref{e:eom1}-\eqref{e:eom2} are given the structure of a stochastic Hamiltonian system. In fact, \cite{H18} treats the case where the fluid's domain is a torus and the perturbation vectors are given by a certain infinite sequence of divergence free vector fields. These are then interpreted as the velocities of molecules which impart their momenta on the sub-blobs. In the present context, because helicity is usually defined on simply-connected domains, the domain is chosen to be $\R^3$. Thus the stochastic perturbation, $\eps\sum e_j\ds W_t^{j,\alpha}$, is interpreted as a model for the combined effect of molecules of different momenta hitting the sub-blob indexed by $\alpha$. Due to the non-compactness of $\R^3$, the Hamiltonian interpretation encounters structural difficulties (Remark~\ref{rem:weak-sp}). Nevertheless, the central (and only) conclusion from the Hamiltonian approach, namely that the momentum map \eqref{e:momap} is constant along solutions, still holds. In fact, Theorem~\ref{thm:prel}(3) shows that \eqref{e:eom2} is equivalent to the preservation of the momentum map. 

Each infinitesimal element, $dx$, is subdivided into a partition of $N$ identical $dx^{\alpha}$, and the initial conditions in each $dx^{\alpha}$ are given by $(g_0^{\alpha}, \cl{\xi_0^{\alpha}}) = (e, \cl{u_0^{\flat}})$, independently of $\alpha$.
As this subdivision is carried out simultaneously for all infinitesimal elements $dx$ in the domain, this may also be seen as $N$ copies of the domain at each point in time. Thus for each $t\in[0,T]$ we have a system of $N$ interacting stochastic fluid states $(g_t^{\alpha},\cl{\xi_t^{\alpha}})$, and, as $N$ becomes large, their average velocity converges to the macroscopically observed state $u_t$, which is the solution to the Navier-Stokes equation~\eqref{e:NS}.  

The construction is Hamiltonian (with the caveat mentioned in Remark~\ref{rem:weak-sp}) and therefore the state of each copy, $(g_t^{\alpha},\xi_t^{\alpha})$, evolves according to a stochastic Hamiltonian system. This does, in general, not mean that energy is conserved (and it is not for the case at hand as discussed in \cite{H18}). However, coadjoint orbits are conserved by stochastic Hamiltonian mechanics when the phase space is the dual of a Lie algebra. Via $d: \gu^*\to\C^2$ the coadjoint orbits in $\gu^*$ are isomorphic to sets of the form $\{\Ad(g)^*\om: g\in\textup{SDiff}(\mathbb{R}^3)\}$. Hence, for each $\alpha$, vorticity is transported along flow lines, i.e.\ \eqref{e:vor-trn} holds.

\section{Helicity representation}\label{sec:hel-rep}
Let $u = u(t,x)$, with $t\in[0,T]$ and $x\in\R^3$, be a solution to the incompressible Navier-Stokes equation. 
The helicity at time $t\in[0,T]$ is 
\begin{equation}
    \label{e:hel}
    \mathcal{H}_t
    = \int_{\mathbb{R}^3}\vv<u_t,\textup{curl}\,u_t>\,dx
    = \int_{\mathbb{R}^3} u_t^{\flat}\wedge\om_t
    = \int_{\mathbb{R}^3} BS^*(\om_t)\wedge\om_t
\end{equation}
where the vorticity, $\om_t = du_t^{\flat}$, and Biot-Savart operator $BS^*$ have been defined in Section~\ref{sec:vor}.  
The physical meaning of the following calculations is discussed in Section~\ref{sec:3phys}.

\subsection{SHIPS representation}\label{sec:hel-rep-ships}
Note that $d: \gu^*\to\mathcal{C}$ has the equivariance property $d\circ\Ad(g)^*=\Ad(g)^*\circ d$, for all $g\in\textup{SDiff}(\mathbb{R}^3)$. 
Define the map
\[
 I: \mathcal{C}^2\to\R,
 \qquad
 \om \mapsto \int_{\mathbb{R}^3} BS^*(\om)\wedge\om.
\]
It follows that $I(\Ad(g)^*\om) = I(g^*\om) = I(\om)$ for all $g\in\textup{SDiff}(\mathbb{R}^3)$. 
Let $(g_t^{\alpha},\om_t^{\alpha})_{\alpha=1}^N$ be a solution to the system \eqref{e:eom1}, \eqref{e:vor1}. Equation~\eqref{e:vor-trn}, which expresses the observation that the vorticity $\om_t^{\alpha}$ is transported along the stochastic flow $g_t^{\alpha}$, implies that 
\begin{align}
    \sum_{\alpha=1}^N I(\Ad(g_t^{\alpha})^*\om_t^{\alpha})
    =
    \sum_{\alpha=1}^N I(\om_0)
    = 
    N\mathcal{H}_0. 
\end{align}

Consider the mean field limit \eqref{e:vor-mf} of the system \eqref{e:eom1}, \eqref{e:vor1}. Hence the (time-dependent and deterministic) vector field $u$, defined in \eqref{e:vor-mf}, satisfies the incompressible Navier-Stokes equation.
Conversely, every solution $u$ can be represented in this manner (Theorem~\ref{thm:prel}).

Define 
\begin{equation}
    \label{e:lambda}
    \lambda_t
    = \lim_{N\to\infty}\sum_{1\le\alpha\neq\beta\le N}
        \Ad\Big( 
            (g_t^{\alpha})^{-1}g_t^{\beta}
        \Big)^* \om_0 / N^2
\end{equation}
where the limit is taken in probability. 

\begin{theorem}\label{thm:hel}
The helicity \eqref{e:hel} of $u$ satisfies
\[
 \mathcal{H}_t
 = 
 \int_{\mathbb{R}^3} BS^*(\lambda_t)\wedge\om_0.
\]
\end{theorem}

\begin{proof}
Indeed, the mean field vorticity formulation \eqref{e:vor-mf} yields
\begin{align*}
    \mathcal{H}_t
    &= 
    \int_{\mathbb{R}^3} BS^*(du_t^{\flat})\wedge du_t^{\flat}
    = 
    \int_{\mathbb{R}^3} BS^*\Big(\lim_{N\to\infty}\sum_{\alpha}\om_t^{\alpha}/N\Big)\wedge \lim_{N\to\infty}\sum_{\beta}\om_t^{\beta}/N \\
    &=
    \lim_{N\to\infty}\Big(
        \underbrace{\sum_{\alpha}\int_{\mathbb{R}^3} BS^*(\om_t^{\alpha})\wedge \om_t^{\alpha}/N^2}_{\mathcal{H}_0/N\longto0}
        +
        \sum_{\alpha\neq\beta}
            \int_{\mathbb{R}^3} BS^*(\om_t^{\alpha})\wedge \om_t^{\beta}/N^2
    \Big)\\
    &=
    \lim_{N\to\infty}
    \sum_{\alpha\neq\beta}\int_{\mathbb{R}^3}
        (g_t^{\beta})^*\,(BS^*(\om_t^{\alpha})\wedge \om_t^{\beta})/N^2\\
    &= 
    \int_{\mathbb{R}^3} BS^*(\lambda_t)\wedge\om_0
\end{align*}
where $\alpha$ and $\beta$ range from $1$ to $N$, and we have used the invariance property of $I$ together with \eqref{e:vor-trn}, which implies
\[
 \Big(g_t^{\beta}\Big)^*\Big(BS^*(\om_t^{\alpha})\wedge \om_t^{\beta}\Big)
 =
 \Big((g_t^{\beta})^*BS^*(\om_t^{\alpha})\Big)\wedge \Big(g_t^{\beta}\Big)^*\om_t^{\beta}
 = 
 \Big(\Ad(g_t^{\beta})^*BS^*(\Ad((g_t^{\alpha})^{-1})^*\om_0)\Big)\wedge \om_0
\]
and 
$\int_{\mathbb{R}^3}(\Ad(g_t^{\beta})^*BS^*(\Ad((g_t^{\alpha})^{-1})^*\om_0))\wedge \om_0 = \int_{\mathbb{R}^3}(BS^*(\Ad(g_t^{\beta})^*\Ad((g_t^{\alpha})^{-1})^*\om_0))\wedge \om_0$;
the last equality holds because 
$d( \Ad(g_t^{\beta})^*BS^*\om_t^{\alpha} - BS^*\Ad(g_t^{\beta})^*\om_t^{\alpha}) = 0$.
\end{proof}

\begin{remark}\label{rem:hel}
The initial vorticity, $\om_0$, is related to the curl of the initial velocity, $u_0$, as $(*\om_0)^{\sharp} = (*du_0^{\flat})^{\sharp} = \textup{curl}\,u_0$.
Further, we have 
$\int_{\mathbb{R}^3}\vv<(*\om_0)^{\sharp},u_0>\,dx 
= \int_{\mathbb{R}^3} \om_0\wedge u_0^{\flat} 
= \int_{\mathbb{R}^3} \vv< \textup{curl}\,u_0, u_0>\,dx
$.
If $X$ is a divergence free vector field its pullback, 
$h^*X = Th^{-1}\cdot(X\circ h) = \Ad(h^{-1})X$, by $h\in\textup{SDiff}(\R^3)$ is again divergence free. 
Hence Theorem~\ref{thm:hel} can be reformulated as
\begin{align*}
 \mathcal{H}_t
 &=
 \int_{\mathbb{R}^3}\Big(BS(*\lam_t)^{\sharp}\Big)\wedge u_0^{\flat}
 =
 \int_{\mathbb{R}^3}\Big\langle BS(*\lam_t)^{\sharp}, \textup{curl}\,u_0\Big\rangle\,dx
 = 
 \int_{\mathbb{R}^3}(**\lam_t)\wedge u_0^{\flat} 
 \\
 &=
 \lim_{N\to\infty}\sum_{1\le\alpha\neq\beta\le N}\int_{\mathbb{R}^3}
    \Ad(g_t^{\beta})^* \Ad((g_t^{\alpha})^{-1})^*\om_0\wedge u_0^{\flat}
    / N^2 \\
 &=
 \lim_{N\to\infty}\sum_{1\le\alpha\neq\beta\le N}\int_{\mathbb{R}^3}
 \Big\langle
    \Ad((g_t^{\beta})^{-1}) \Ad(g_t^{\alpha})\,\textup{curl}\, u_0, 
    u_0 
 \Big\rangle\,dx
    / N^2 \\
 &=
 \lim_{N\to\infty}\sum_{1\le\alpha\neq\beta\le N}\int_{\mathbb{R}^3}
    \Big\langle
     \Ad(g_t^{\alpha})^{\top} \Ad((g_t^{\beta})^{-1})^{\top} \,u_0 ,
    \textup{curl}\,u_0 
    \Big\rangle\,dx / N^2 .
\end{align*}
\end{remark}

\subsection{Mean field limit}\label{sec:hel-rep-mf}
To find the evolution equation characterizing the limit $\lam_t$ defined in \eqref{e:lambda}, set $h_t^{\beta,\alpha} = (g_t^{\beta})^{-1} g_t^{\alpha}$ and
\begin{equation}
    \eta_t^{\beta,\alpha}
    = \Ad\Big( (g_t^{\alpha})^{-1}g_t^{\beta} \Big)^*\om_0 
    = \Ad\Big( (h_t^{\beta,\alpha})^{-1}\Big)^*\om_0.
\end{equation}
The product formula for Stratonovich equations implies that the process $h_t^{\beta,\alpha}$ in $\textup{SDiff}(\R^3)$ is the solution  to
\begin{align}
 \notag 
    \ds h_t^{\beta,\alpha}
    &=
    TR^{h_t^{\beta,\alpha}}\cdot
    \Big( 
    -(Tg_t^{\beta})^{-1}
    \underbrace{(\ds g_t^{\beta})(g_t^{\beta})^{-1}}_{\eqref{e:eom1}}
    \underbrace{g_t^{\alpha}(h_t^{\beta,\alpha})^{-1}}_{g_t^{\beta}}
    + 
    (Tg_t^{\beta})^{-1}
    \underbrace{(\ds g_t^{\alpha})(g_t^{\alpha})^{-1}}_{\eqref{e:eom1}}
    g_t^{\beta} 
 \Big) \\
\label{e:h^ab}
  &= 
  \eps \, TR^{h_t^{\beta,\alpha}}\cdot
  \Big( \Ad(g_t^{\beta})^{-1} \ds \hat{W}^{\alpha,\beta}
  \Big) \\
  \notag
  h_0^{\beta,\alpha} 
  &= e
\end{align}
where $\hat{W}^{\alpha,\beta} = \sum e_j(W^{j,\alpha}-W^{j,\beta})$ is a difference of two independent Brownian motions. It is therefore, up to a multiplicative factor, again a Brownian motion with quadratic variation $[\hat{W}^{\alpha,\beta},\hat{W}^{\alpha,\beta}]_t = [W^{\alpha},W^{\alpha}]_t + [W^{\beta},W^{\beta}]_t = 2t$.

Equation~\eqref{e:h^ab} implies that $\eta_t^{\beta,\alpha}$ satisfies
\begin{align}
\label{e:etaSDE}
 \ds\eta_t^{\beta,\alpha}
 = 
 -\eps\,
 \ad\Big( \Ad(g_t^{\beta})^{-1} \ds \hat{W}^{\alpha,\beta}
 \Big)^* \eta_t^{\beta,\alpha},
 \qquad
 \eta_0^{\beta,\alpha}  = \om_0.
\end{align}

Equations~\eqref{e:h^ab} and \eqref{e:etaSDE} are SDEs with random coefficients, since there is a dependence on realizations of $g_t^{\beta}$. But neither depend on $g_t^{\alpha}$ for $\alpha\neq\beta$. 
Therefore, for all $N$ and all $\beta\le N$, the sequence $\eta_t^{\beta,\alpha}$ with  $\alpha=1,\ldots,\hat{\beta},\ldots,N$ ($\beta$ omitted) is i.i.d., and we have that 
\begin{equation}\label{e:iid}
 \sum_{\alpha\neq\beta, \alpha=1}^N\eta_t^{\beta,\alpha}
 \sim 
 (N-1)\eta_t^{\beta,\beta+1}
\end{equation}
where $\sim$ means equivalence in distribution. (For $\beta= N$ the expression $\eta^{\beta,\beta+1}$ does not make sense, thus one should write, e.g., $\sum_{\alpha=1}^{N-1}\eta_t^{N,\alpha} \sim  (N-1)\eta_t^{N,N-1}$ for this case. However, below we will only need the case $\beta=1$ and so the inconsistency at $\beta=N$ will be ignored from now on.) 

Recall from Remark~\ref{rem:N} that $g_t^{\beta}$ depends on $N$. 
Let $N$ go to infinity and assume $g_t = \lim_{N\to\infty} g_t^{\beta}$  is the stochastic mean field limit (for an arbitrarily fixed $\beta$, e.g.\ $\beta=1$) of the IPS \eqref{e:eom1}, given by the mean field SDE~\eqref{e:mf-eom1}
and where the driving Brownian motion is denoted by $W$.
 
Equations~\eqref{e:h^ab} and \eqref{e:etaSDE} imply, respectively,
that $h_t = \lim_{N\to\infty}h_t^{\beta,\beta+1}$ satisfies 
\begin{equation}
    \label{e:h}
    \ds h_t 
    =   
    \eps \, TR^{h_t}\cdot
    \Big( \Ad(g_t)^{-1} \ds \hat{W}
      \Big) ,
  \qquad
  h_0
  = e
\end{equation}
and
that $\eta_t = \Ad(h_t^{-1})^*\om_0 = \lim_{N\to\infty}\eta_t^{\beta,\beta+1}$ satisfies 
\begin{equation}
    \label{e:etaSDE2}
    \ds \eta_t 
    = -\eps\, \ad\Big(\Ad(g_t)^{-1}\ds \hat{W}\Big)^*\eta_t,
    \quad 
    \eta_0 = \om_0
\end{equation}
where $\hat{W} = B-W$ and $B$ is a Brownian motion in $\R^3$ independent of $W$. 

Now, to find the evolution equation for \eqref{e:lambda}, note that \eqref{e:iid} yields
\begin{align}
\label{e:lam}
    \lambda_t
    = 
    \lim_{N\to\infty}\sum_{\beta=1}^N\sum_{\alpha\neq\beta, \alpha=1}^N\eta_t^{\beta,\alpha}/N^2
    = 
    \lim_{N\to\infty}\sum_{\beta=1}^N(N-1)\eta_t^{\beta,\beta+1}/N^2
    =
    \lim_{N\to\infty}\sum_{\beta=1}^N \eta_t^{\beta,\beta+1}/N 
    = E\Big[\eta_t\Big].
\end{align}
Here we use the propagation of chaos property of mean field limits (\cite{Oel84,JW17}) which implies independence of $g_t^{\beta}$ in the limit as $N\to\infty$, such that the $\eta_t^{\beta,\beta+1}$ are asymptotically i.i.d.

However, because \eqref{e:etaSDE2} is an SDE with random coefficients, it does not have the Markov property and one cannot expect the evolution of $\lambda_t = E[\eta_t]$ to be given by a deterministic PDE. 

\begin{theorem}
\label{thm:hel-mf}
The helicity \eqref{e:hel} satisfies 
\begin{align*}
    \mathcal{H}_t 
    = 
    \int_{\mathbb{R}^3} BS^*(E[\eta_t])\wedge\om_0 
    =
    E\Big[\int_{\mathbb{R}^3} \Big(\Ad(h_t^{-1})^*\om_0\Big) \wedge u_0^{\flat}\Big] 
    = 
    E\Big[\int_{\mathbb{R}^3}
        \vv< \Ad(h_t)^{\top}u_0, \,\textup{curl}\,u_0>\,dx \Big].
\end{align*}
Moreover, 
\begin{equation}
    \by{\del}{\del t} E [\eta_t]
    = \eps^2 E\Big[
            \Ad(g_t)^*\Delta \Ad(g_t^{-1})^*\eta_t 
        \Big]
\end{equation}
where $\Delta = (d+*d*)^2$ is the Laplacian. 
\end{theorem}

\begin{proof}
The first part follows from \eqref{e:lam} and Theorem~\ref{thm:hel}, and because
$\int_{\mathbb{R}^3} (\Ad(h_t^{-1})^*\om_0) \wedge u_0^{\flat}
=
\int_{\mathbb{R}^3} *\,(\textup{curl}\,u_0)^{\flat} \wedge \Ad(h_t)^* u_0^{\flat}
=
\int_{\mathbb{R}^3} \vv<\Ad(h_t)\,\textup{curl}\,u_0 , u_0>\,dx
$.
For the second statement, it remains to transform \eqref{e:etaSDE2} into Ito form. 
By definition of the Stratonovich integral (see \cite{Pro}) it follows that
\begin{align*}
    \eta_t-\eta_0
    &= 
    -\eps\sum_{j=1}^3\int_0^t\ad\Big(\Ad(g_s)^{-1}e_j\Big)^*\eta_s\ds\hat{W}_s^{j}\\
    &=
    -\eps\sum_{j=1}^3\int_0^t\ad\Big(\Ad(g_s)^{-1}e_j\Big)^*\eta_s\di\hat{W}_s^{j}
    -\frac{\eps}{2}\sum_{j=1}^3\Big[
        \ad\Big(\Ad(g_.)^{-1}e_j\Big)^*\eta_., \hat{W}_.^j 
    \Big]_t 
\end{align*}
where $\di$ indicates Ito differentiation and $[.,.]_t$ is the quadratic variation process. 
The product formula for Stratonovich SDEs applied to equations~\eqref{e:eom1} and \eqref{e:etaSDE2} implies
\begin{equation}
    \label{e:pf3}
    \ds\Big(\Ad(g_t^{-1})^*\eta_t\Big)
    = -\ad\Big(u\ds t + \eps\ds B_t\Big)^*\Ad(g_t^{-1})^*\eta_t
\end{equation}
whence
\begin{align*}
    \ds\Big( \ad(\Ad(g_t^{-1})e_j)^*\eta_t \Big)
    &=
    \ds\Big( \Ad(g_t)^*\ad(e_j)^*\Ad(g_t^{-1})^*\eta_t \Big)
    \\
    &=
    \Big(\ldots\Big)\ds t
    +
    \eps\sum_{k=1}^3\Ad(g_t)^*\ad(e_k)^*\ad(e_j)^*\Ad(g_t^{-1})^*\eta_t\ds W^k\\
    &\phantom{==}
    -
    \eps\sum_{l=1}^3\Ad(g_t)^*\ad(e_j)^*\ad(e_l)^*\Ad(g_t^{-1})^*\eta_t\ds B^k.
\end{align*}
Using \cite[Ch.~2,~Thm.~29]{Pro},
\begin{align*}
    \Big[
        \ad\Big(\Ad(g_.^{-1})e_j\Big)^*\eta_., \hat{W}_.^j 
    \Big]_t
    &=
    \eps\sum_{k=1}^3\Big[ 
        \int_0^.\Ad(g_s)^*\ad(e_k)^*\ad(e_j)^*\Ad(g_s^{-1})^*\eta_.\di W_s^k, -W^j_.
    \Big]_t\\
    &\phantom{==}
    -
    \eps\sum_{l=1}^3\Big[
        \int_0^.\Ad(g_s)^*\ad(e_j)^*\ad(e_l)^*\Ad(g_s^{-1})^*\eta_.\di B_s^l, B^j_.
    \Big] \\
    &= 
    -2\eps \int_0^t
    \Ad(g_s)^*\ad(e_j)^*\ad(e_j)^*\Ad(g_s^{-1})^*\eta_s\,ds.
\end{align*}
Therefore,
\begin{equation}
    \label{e:eta-ito}
    \di\eta_t
    = 
    \eps^2\Ad(g_t)^*\Delta\Ad(g_t^{-1})^*\eta_t\di t
    - 
    \eps\, \ad\Big(\Ad(g_t)^{-1}\di \hat{W}_t\Big)^*\eta_t
\end{equation}
and the claim follows since 
$E[\int_0^t \ad(\Ad(g_s)^{-1}\di \hat{W}_s)^*\eta_s] = 0$.
\end{proof}

\begin{remark}
Equation~\eqref{e:pf3} has the same structure and initial condition as \eqref{e:vor-mf}. But the driving Brownian motions are different, thus these equations do not imply path-wise equality of $\Ad(g_t^{-1})^*\eta_t$ and $\om_t$.  
\end{remark}

The transpose of the adjoint operator, $\ad(\cdot)^{\top}: \gu\to\gu$, is characterized by $\ww<\ad(X)^{\top}Y, Z> = \ww<Y, \ad(X)Z>$ for $X,Y\in\gu$ and $Z\in\gu_0$, and given by 
\[
 \ad(X)^{\top}Y
 = P\Big(\nabla_X Y + (\nabla^{\top} X)Y \Big)
\]
where $(\nabla^{\top} X)Y = \sum (\del_i X^j)Y^j e_i$.
Consider now a solution, $\cl{\xi_t}$,  to the mean field equation \eqref{e:mf-eom2} and let the corresponding vector field valued process be defined by $u_t^W = P\xi_t^{\sharp}$.  Then $u_t^W$ satisfies 
\[
 \ds u_t^W 
 = -\ad\Big( u_t \ds t + \eps\ds W_t \Big)^{\top}u_t^W 
 = -P\Big(\nabla_{u_t}u_t^W + (\nabla^{\top}u_t)u_t^W\Big)\ds t
  - \sum_{k=1}^3P \nabla_{e_k}u_t^W \ds W_t^k
\]
with $u_t = E[u_t^W]$ and where $W$ is the same Brownian motion as in \eqref{e:mf-eom1}. Consider furthermore $u_t^B$ defined by 
\[
 \ds u_t^B 
 = -\ad\Big( u_t \ds t + \eps\ds B_t \Big)^{\top}u_t^B 
 = -P\Big(\nabla_{u_t}u_t^B + (\nabla^{\top}u_t)u_t^B\Big)\ds t
  - \sum_{k=1}^3P \nabla_{e_k}u_t^B \ds B_t^k
\]
with $u_t = E[u_t^B] = E[u_t^W]$ and where $B$ is the same as in \eqref{e:etaSDE2}. 
Comparing equations~\eqref{e:vor-mf} and \eqref{e:pf3} then implies $\textup{curl}\,u_t^B = (*\Ad(g_t^{-1})^*\eta_t)^{\sharp}$. 
Since $\eps^2 = 2\nu$ Theorem~\ref{thm:hel-mf} implies in particular  that
\begin{align}
\label{e:inf}
 \by{\del}{\del t}\mathcal{H}_t
 &= 
    \eps^2 E\Big[ \int_{\mathbb{R}^3}
    \Big(\Ad(g_t)^*\Delta \Ad(g_t^{-1})^*\eta_t\Big)\wedge BS^* \om_0
    \Big] 
 =
    \eps^2 E\Big[ \int_{\mathbb{R}^3}
     (\Delta \underbrace{\Ad(g_t^{-1})^*\eta_t}_{*(\textup{curl}\,u_t^B)^{\flat}} )
     \wedge 
     \underbrace{\Ad(g_t^{-1})^*\xi_0}_{ (u_t^W)^{\flat} }
    \Big] \\
\notag
 &= 
 \eps^2 E\Big[ 
  \int_{\mathbb{R}^3}\vv<\textup{curl}\,\Delta\,u_t^B,  u_t^W>\,dx
 \Big] 
 =
 \eps^2  
  \int_{\mathbb{R}^3}\vv<\textup{curl}\,\Delta\, E[u_t^B],  E[u_t^W]>\,dx
 =
 2\nu 
  \int_{\mathbb{R}^3}\vv<\textup{curl}\,\Delta\,u_t,  u_t>\,dx
\end{align}
where we use $*\Delta = \Delta *$ and that $u^W$ and $u^B$ are independent.
This coincides, of course, with the result of using the Navier-Stokes equation to obtain $\by{\del}{\del t}\mathcal{H}_t$ from the definition~\eqref{e:hel}.

\subsection{Physical interpretation}\label{sec:3phys}
The proof of Theorem~\ref{thm:hel} offers two representations for the helicity $\mathcal{H}_t$: 

The first is 
\begin{equation}
    \label{e:3p1}
    \mathcal{H}_t
    = \lim_{N\to\infty}\sum_{1\le\alpha\neq\beta\le N}\int_{\mathbb{R}^3} BS^*(\om_t^{\alpha})\wedge\om_t^{\beta} / N^2 . 
\end{equation}
Each integral of the form $\int_{\mathbb{R}^3} BS^*(\om_t^{\alpha})\wedge\om_t^{\beta}$ represents the overall linking of $(*\om_t^{\alpha})^{\sharp}$-lines to $(*\om_t^{\beta})^{\sharp}$-lines (\cite{M69}). Let us refer to this quantity as simply the linking of $\om_t^{\alpha}$ and $\om_t^{\beta}$. 
Then \eqref{e:3p1} says that helicity can be interpreted as the average linking of $\om_t^{\alpha}$ to $\om_t^{\beta}$
(over all pairs $\alpha, \beta$ with $\alpha\neq\beta$). 
Because each vorticity, $\om_t^{\alpha}$, is transported along its own stochastic Lagrangian path, $g_t^{\alpha}$, the self-linking of $\om_t^{\alpha}$ remains constant. This is a consequence of \eqref{e:vor-trn}, which implies that $\int_{\mathbb{R}^3} BS^*(\om_t^{\alpha})\wedge\om_t^{\alpha} = \mathcal{H}_0$ for all $\alpha$. However, as $\om_t^{\alpha}$ and $\om_t^{\beta}$ are driven by different Brownian motions for $\alpha\neq\beta$, the linking of their respective vortex lines can change as time progresses.  
 
The second representation of helicity, compare also with Remark~\ref{rem:hel}, is 
\begin{equation}
    \label{e:3p2}
    \mathcal{H}_t
    = \int_{\mathbb{R}^3} BS^*(\lam_t)\wedge\om_0
    =
 \lim_{N\to\infty}\sum_{1\le\alpha\neq\beta\le N}\int_{\mathbb{R}^3}
    \Big\langle
         \Ad(g_t^{\alpha})^{\top}\Ad((g_t^{\beta})^{-1})^{\top}\,u_0 ,
       \,\textup{curl}\,u_0 
    \Big\rangle\,dx \, / N^2.
\end{equation}
This arises  because of the transport property~\eqref{e:vor-trn}. Indeed, the stochastic flow, $g_t^{\beta}$, lies in the group of volume preserving diffeomorphisms, and this allows to pull-back $\om_t^{\alpha}$ along $g_t^{\beta}$. The result is \eqref{e:3p2}, which is now the average overall linking of the initial vorticity, $\om_0$, to the vorticities associated to backward-forward transports of $u_0$. 

The second interpretation leads to the mean field limit $\lam_t = E[\eta_t]$ 
This shows that $\mathcal{H}_t$ equals the expectation of the cross-helicity of the stochastically transported initial condition, $\Ad(h_t)^{\top}u_0$,  and $\textup{curl}\,u_0$.



\end{document}